\documentclass[runningheads,a4paper]{llncs}

\usepackage{epsfig}
\usepackage{graphics}
\usepackage{amsmath,setspace}

\usepackage[dvipsnames,usenames]{color}
\usepackage[colorlinks=true,urlcolor=Blue,citecolor=Green,linkcolor=BrickRed]{hyperref}
\newcommand{\x}{{x}}

\newcommand{\G}{{\mathcal G}}

\begin{document}
\mainmatter
%
%
\title{Approximating the Diameter of Planar Graphs\\ in Near Linear Time}

\author{Oren Weimann \and Raphael Yuster}

\institute{University of Haifa, Israel\\ \href{mailto:oren@cs.haifa.ac.il}{oren@cs.haifa.ac.il},
 \href{mailto:raphy@math.haifa.ac.il}{raphy@math.haifa.ac.il}}

\date{}
\maketitle

\begin{abstract}
We present a $(1+\varepsilon)$-approximation algorithm running in ${O}(f(\varepsilon) \cdot n \log^4 n)$ time for finding the diameter of an undirected planar graph with $n$ vertices and with non-negative edge lengths. 
\end{abstract}
%

\section{Introduction}

The diameter of a graph is the largest distance between two vertices. Computing it is among the most fundamental algorithmic graph problems. 
%
%
In general weighted graphs, as well as in planar graphs, the only known way to compute the diameter is to essentially solve the (more general) All-Pairs Shortest Paths (APSP) problem and then take the pair of vertices with the largest distance.

In general weighted graphs with $n$ vertices and $m$ edges, solving APSP (thus diameter) currently requires $\tilde O(n^3)$ time. The fastest algorithm to date is $O(n^3(\log\log n)/ \log^2n)$ by Han and Takaoka~\cite{Han}, or for sparse graphs $O(mn + n^2\log n)$ by Johnson~\cite{Johnson}, with a small improvement to $O(mn + n^2\log\log n)$  \cite{pettie-2002}.

In weighted {\em planar} graphs, solving APSP can be done in $O(n^2)$ time by Frederickson~\cite{Frederickson}. While this is  optimal for APSP, it is not clear that it is optimal for diameter. Currently, only a logarithmic factor improvement by Wulff-Nilsen~\cite{WN-diam} is known for the diameter, running in  $O(n^2(\log\log n)^4/ \log n)$ time. A long standing open problem~\cite{Chung}  is to find the diameter in truly  subquadratic $O(n^{2-\varepsilon})$ time. 
Eppstein~\cite{Eppstein} has shown that if the diameter in a planar graph is bounded by a fixed constant, then it can be found in $O(n)$ time.
Fast algorithms are also known for some simpler classes of graphs like outer-planar graphs~\cite{16}, interval graphs~\cite{23}, and others~\cite{Chordal,13}.

In lack of truly subcubic-time algorithms for general graphs and truly subquadratic   time algorithms for planar graphs it is natural to seek faster algorithms that {\em approximate} the diameter.  
It is easy to approximate the diameter within a factor of 2 by simply computing a Single-Source Shortest Path (SSSP) tree from any vertex in the graph and returning
twice the depth of the deepest node in the tree.  This requires $O(m+n\log n)$ time for general graphs and $O(n)$ time for planar graphs~\cite{HKRS97}. For general graphs, Aingworth et al.~\cite{Aingworth} improved the approximation factor from 2 to 3/2 at the cost of $\tilde{O}(m\sqrt{n}+n^2)$ running time, and 
Boitmanis et al.~\cite{5} gave an additive approximation factor of $O(\sqrt{n})$ with  $\tilde O(m\sqrt{n})$ running time. For planar graphs, the current best approximation is a 3/2-approximation by Berman and Kasiviswanathan running in $O(n^{3/2})$ time~\cite{WADS}. We improve this to a $(1+\varepsilon)$-approximation running in  $\tilde O(n)$ time for any fixed $0<\varepsilon < 1$. More precisely, we prove the following theorem:

\begin{theorem}\label{theorem:main}
Given an undirected planar graph with $n$ vertices, non-negative edge lengths, and diameter $d$. For any $\varepsilon >0$ we can compute an approximate diameter $d'$ (where $d \le d'  \le (1+\varepsilon)\cdot d$) in time $O(n \log^4 n/\varepsilon^4 +n \cdot2^{O(1/\varepsilon)})$.
\end{theorem}

\paragraph{\bf Summary of the Algorithm.}

A lemma of Lipton and Tarjan~\cite{LiptonTarjan79} states that, for any SSSP tree $T$ in a planar graph $G$, there is a non-tree edge $e$ (where $e$ might possibly be a
non-edge of the planar graph) such that the 
strict interior and strict exterior of the unique simple cycle $C$ in $T\cup \{e\}$ each contains at most $2/3 \cdot n$ vertices.
The vertices of $C$ therefore form a {\em separator} consisting of two shortest paths with the same common starting vertex.

Let $G_{in}$ (resp. $G_{out}$) be the subgraph of $G$ induced by $C$ and all interior (resp. exterior) vertices to $C$. Let  $d(G_{in},G_{out},G)$ denote the largest distance in the graph $G$ between a {\em marked} vertex in $V(G_{in})$ and a {\em marked} vertex in $V(G_{out})$. In the beginning, all vertices of $G$ are marked and we seek the diameter which is $d(G,G,G)$. We use a divide and conquer algorithm that first approximates $d(G_{in},G_{out},G)$, then unmarks all vertices of $C$, and then recursively approximates $d(G_{in},G_{in},G)$ and $d(G_{out},G_{out},G)$ and takes the maximum of all three. We outline this algorithm below. Before running it, we compute an SSSP tree from any  vertex using the linear-time SSSP algorithm of Henzinger et al.~\cite{HKRS97}. The depth of the deepest node in this tree already gives a 2-approximation to the diameter $d(G,G,G)$. Let $\x$ be the obtained value such that $\x \le d(G,G,G) \le 2\x$.

\paragraph{Reduce $d(G_{in},G_{out},G)$ to $d(G_{in},G_{out}, G_t)$ in a tripartite graph $G_t$:}
The separator $C$ is composed of two shortest paths $P$ and $Q$ emanating from the same vertex, but that are otherwise disjoint. 
We carefully choose a subset of  $16/\varepsilon$ vertices from $C$ called {\em portals}. The first (resp. last) $8/\varepsilon$ portals are  all part of the prefix of $P$ (resp. $Q$) that is of length $8\x$. 
The purpose of the portals is to approximate a shortest $u$-to-$v$ path for $u \in G_{in}$ and $v \in G_{out}$ by forcing it to go through a portal. Formally, we construct a 
tripartite graph $G_t$ with vertices $(V(G_{in}), portals, V(G_{out}))$. The length of edge $(u\in V(G_{in}),v\in portals)$ or $(u\in portals,v\in V(G_{out}))$ in  $G_t$ is the $u$-to-$v$ distance in $G$. This distance is computed by running the SSSP algorithm of~\cite{HKRS97} from each of the $16/\varepsilon$ portals. By the choice of portals, we show that 
 $d(G_{in},G_{out},G_t)$  is  a $(1+2\varepsilon)$-approximation of $d(G_{in},G_{out},G)$.

\paragraph{Approximate $d(G_{in},G_{out},G_t)$:}
If $\ell$ is the maximum edge-length of $G_t$, then note that $d(G_{in},G_{out},G_t)$ is between $\ell$ and $2\ell$.
This fact makes it possible to round the edge-lengths of $G_t$ to be in $\{1, 2, \ldots, 1/\varepsilon\}$ so that $\varepsilon \ell \cdot d(G_{in},G_{out},G_t)$ after rounding is a $(1+2\varepsilon)$-approximation  to $d(G_{in},G_{out},G_t)$ before rounding. For any fixed $\varepsilon$ we can assume without loss of generality that $1/\varepsilon$ is an integer. This means that after rounding $d(G_{in},G_{out},G_t)$ is bounded by some fixed integer.
We give a linear-time algorithm to compute it exactly, thus approximating $d(G_{in},G_{out},G)$. We then unmark all vertices of $C$ and move on to recursively approximate $d(G_{in},G_{in},G)$ (the case of $d(G_{out},G_{out},G)$ is symmetric).
 
 \paragraph{Reduce  $d(G_{in},G_{in},G)$ 
to $d(G_{in},G_{in},G_{in}^+)$  in a planar graph $G_{in}^+$ of size at most $ 2/3 \cdot n$:} 
In order to apply recursion, we construct planar graphs $G_{in}^+$ and  $G_{out}^+$ (that is constructed similarly to $G_{in}^+$). The size of each of these graphs will be at most $2/3 \cdot n$ and their total size  $n+o(n)$.
We would like $G_{in}^+$ to be such that $d(G_{in},G_{in},G_{in}^+)$  is a 
 $(1+\varepsilon/(2\log n))$-approximation\footnote{$\log n =\log_2 n$ throughout the paper.} to {$d(G_{in},G_{in},G)$. 

To construct $G_{in}^+$, we first choose a subset of  $256\log n/\varepsilon$ vertices from $C$ called {\em dense portals}. We then compute all  $O((256\log n/\varepsilon)^2)$ shortest paths in $G_{out}$ between dense portals. The graph $B'$ obtained by the union of all these paths has at most $O((256\log n/\varepsilon)^4)$ vertices of degree $> 2$. We contract vertices of degree $= 2$ so that the number of vertices in  $B'$ decreases to $O((256\log n/\varepsilon)^4)$. Appending this small graph $B'$ (after unmarking all of its vertices) as an exterior to $G_{in}$ results in a graph $G_{in}^+$ that has $|G_{in}|+O((256\log n/\varepsilon)^4)$ vertices and $d(G_{in},G_{in},G_{in}^+)$ is a $(1+\varepsilon/(2\log n))$-approximation of $d(G_{in},G_{in},G)$. 

The problem is still that  the size of $G_{in}^+$ is not necessarily bounded by $2/3 \cdot n$. This is because $C$ (that is part of $G_{in}^+$) can be as large as $n$. 
We show how to shrink $G_{in}^+$ to size roughly $2/3 \cdot n$ while $d(G_{in},G_{in},G_{in}^+)$ remains a $(1+\varepsilon/(2\log n))$-approximation of $d(G_{in},G_{in},G)$. 
To achieve this, we shrink the $C$ part of $G_{in}^+$ so that it only includes the dense portals without changing $d(G_{in},G_{in},G_{in}^+)$. 

\paragraph{Approximate $d(G_{in},G_{in},G_{in}^+)$:}
Finally, once $|G_{in}^+|\le 2/3 \cdot n$ we  apply recursion to $d(G_{in},G_{in},G_{in}^+)$. In the halting condition, when $|G_{in}^+|\le (256\log n /\varepsilon)^4$, we naively compute $d(G_{in},G_{in},G_{in}^+)$ using APSP.

\paragraph{\bf Related Work.}
The use of shortest-path separators and portals to approximate distances in planar graphs was first suggested in the context of {\em approximate distance oracles}. These are data structures that upon query $u,v$ return a  $(1 + \varepsilon)$-approximation of the $u$-to-$v$ distance. 
Thorup~\cite{Thorup04} presented an $O(1 /\varepsilon \cdot n \log n)$-space oracle answering queries in $O(1/\varepsilon)$ time on directed weighted planar graphs. Independently, Klein~\cite{Klein2002} achieved these same bounds for undirected planar graphs. 

In distance oracles, we need distances between every pair of vertices and each vertex is associated with a possibly different set of portals. In our diameter case however, since we know the diameter is between $x$ and $2x$, it is possible to associate all vertices with the exact same set of portals. This fact is crucial in our algorithm, both for its running time and for its use of rounding.
Another important distinction between our algorithm and distance oracles is that  distance oracles upon query $(u,v)$ can inspect all recursive subgraphs that include both $u$ and $v$. We on the other hand must have that, for every $(u,v)$, the shortest $u$-to-$v$ path exists (approximately) in the unique subgraph where $u$ and $v$ are separated by $C$. This fact necessitated our construction  of  $G_{in}^+$ and $G_{out}^+$.

\section{The Algorithm}

In this section we give a detailed description of an algorithm that approximates the diameter of an undirected weighted planar graph $\G=(V,E)$ in the bounds of Theorem~\ref{theorem:main}. The algorithm computes a $(1+\varepsilon)$-approximation of the diameter $d=d(G,G,G)$ for $G=\G$. This means it returns a value $d'$ where  $d\le d'  \le (1+\varepsilon)\cdot d$  (recall that, before running the algorithm, we compute a value $\x$ such that $ \x \le d \le 2\x$ by computing a single-source shortest-path tree from an arbitrary vertex in $\G$). 
We focus on approximating the value of the diameter. An actual path of length $d'$ can be found in the same time bounds. For simplicity we will assume that shortest paths are unique. This can always be achieved by adding random infinitesimal
weights to each edge, and can also be achieved deterministically using lexicographic-shortest paths (see, e.g.,~\cite{HartvigsenMardon}).
Also, to simplify the presentation, we  assume that  $\varepsilon \le 0.1$ and we describe a  $(1+7\varepsilon)$-approximation. (then just take $\varepsilon' = \varepsilon/7$).

The algorithm is recursive and actually solves the more general problem of finding the largest distance only between all pairs of {\em marked vertices}. In the beginning, we mark all $n=|V(G)|$ vertices of $G=\G$ and set out to approximate $d(G,G,G)$ (the largest distance in $G$ between marked vertices in $V(G)$).
Each recursive call approximates the largest distance in a specific subset of marked vertices, and then unmarks some vertices before the next recursive call. We make sure that whatever the endpoints of the actual diameter are, their distance is approximated in some recursive call. Finally, throughout the recursive calls, we maintain the invariant that the distance between any two {\em marked} vertices in the graph $G$ of the recursive call is a $(1+\varepsilon)$-approximation of their distance in the original graph $\G$ (there is no guarantee on the marked-to-unmarked or the unmarked-to-unmarked distances).
 We denote by $\delta_{\G}(u,v)$ the $u$-to-$v$ distance in the original graph $\G$. 
 
The recursion is applied according to a variant of the shortest-path separator decomposition for
planar graphs by Lipton and Tarjan~\cite{LiptonTarjan79}: We first pick any {\em marked} vertex $v_1$ and compute in linear time the SSSP tree from $v_1$ in $G$. In this tree, we can find in linear time two shortest paths $P$ and $Q$ (both emanating from $v_1$) such that removing the vertices of  $C= P \cup Q$ from $G$
results in  two disjoint planar subgraphs $A$ and $B$ (i.e., there are no edges in $V(A)\times V(B)$).
The number of vertices of $C$ can be as large as $n$ but it is guaranteed that $|V(A)|\le 2/3 \cdot n$ and $|V(B)|\le 2/3 \cdot n$. 
Notice that the paths $P$ and $Q$ might share a common prefix. It is common to not include this shared prefix in $C$. However, in our case, we must have the property that $P$ and $Q$ start at a {\em marked} vertex. So we  include in $C$ the shared prefix as well. See Fig.~\ref{separator} (left).

Let $G_{in}$  (resp. $G_{out}$) be the subgraph of $G$ induced by $V(C) \cup V(A)$  (resp. $V(C) \cup V(B)$). In order to approximate $d(G,G,G)$, we first compute a $(1+5\varepsilon)$-approximation $d_1$ of $d(G_{in},G_{out},G)$ (the largest distance in $G$ between the marked vertices of $V(G_{in})$ and the marked vertices of $V(G_{out})$).
In particular,  $d_1$  takes into account all $V(C) \times V(G)$ distances.
We can therefore unmark all the vertices of $C$ and
move on to approximate $d_2= d(G_{in},G_{in},G)$ (approximating  $d_3 = d(G_{out},G_{out},G)$ is done similarly). 
We approximate $d(G_{in},G_{in},G)$ by applying recursion on $d(G_{in},G_{in},G_{in}^+)$ where $|V(G_{in}^+)| \le 2/3 \cdot n$. The marked vertices in $G_{in}^+$ and in $G_{in}$ are the same and $d(G_{in},G_{in},G_{in}^+)$ is a $(1+\varepsilon/(2\log n))$-approximation of $d(G_{in},G_{in},G)$. This way, the diameter grows by a multiplicative factor of $(1+ \varepsilon/(2\log n))$  in each recursive call. Since the recursive depth is $O(\log n)$ (actually, it is never more than $1.8\log n$) we get a  $(1+5\varepsilon)\cdot (1+\varepsilon)\le (1+7\varepsilon)$-approximation $d_2$ to $d(G_{in},G_{in},G)$.
Finally, we return  $\max\{d_1,d_2,d_3\}$.

\subsection{Reduce $d(G_{in},G_{out},G)$ to $d(G_{in},G_{out}, G_t)$}
Our goal is now to approximate $d(G_{in},G_{out},G)$. For $u\in G_{in}$ and $v\in G_{out}$, we approximate a shortest $u$-to-$v$ path in $G$  by forcing it to go through a {\em portal}. In other words, consider a shortest $u$-to-$v$ path. It is is obviously composed of a shortest $u$-to-$c$ path in $G$ concatenated with a shortest $c$-to-$v$ path in $G$ for some vertex $c \in C$. We approximate the shortest $u$-to-$v$ path by insisting that $c$ is a portal. The fact that we only need to consider $u$-to-$v$ paths that are of length between $\x$ and $2\x$ makes it possible to choose the same portals for all vertices.

We now describe how to choose the portals in  linear time. Recall that the separator $C$ is composed of two shortest paths $P$ and $Q$ emanating from the same {\em marked} vertex $v_1$. 
The vertex $v_1$ is chosen as the first  portal. Then, for $i=2,3,\ldots$ we start from $v_{i-1}$ and walk on $P$ until we reach the first vertex $v$ whose distance from $v_{i-1}$ via $P$ is greater than $\varepsilon \x$. We designate $v$ as the portal $v_i$ and continue to $i+1$. We stop the process when we encounter a vertex $v$ whose distance from $v_1$ is greater than $8\x$. This guarantees that at most $8/\varepsilon$ portals are chosen from the shortest path $P$ and they are all in a prefix of $P$ of length at most $8\x$. This might seem counterintuitive as we know that any shortest path $P$ in the original graph $\G$ is of length at most $2\x$. However, since one endpoint of $P$ is not necessarily marked, it is possible that $P$ is a shortest path in $G$ but not even an approximate shortest path in the original graph $\G$. We do the same for $Q$, and we get a total  of $16/\varepsilon$ portals. See Fig.~\ref{separator} (right).

\begin{figure}[h!]
   \centering
   \includegraphics[scale=0.18]{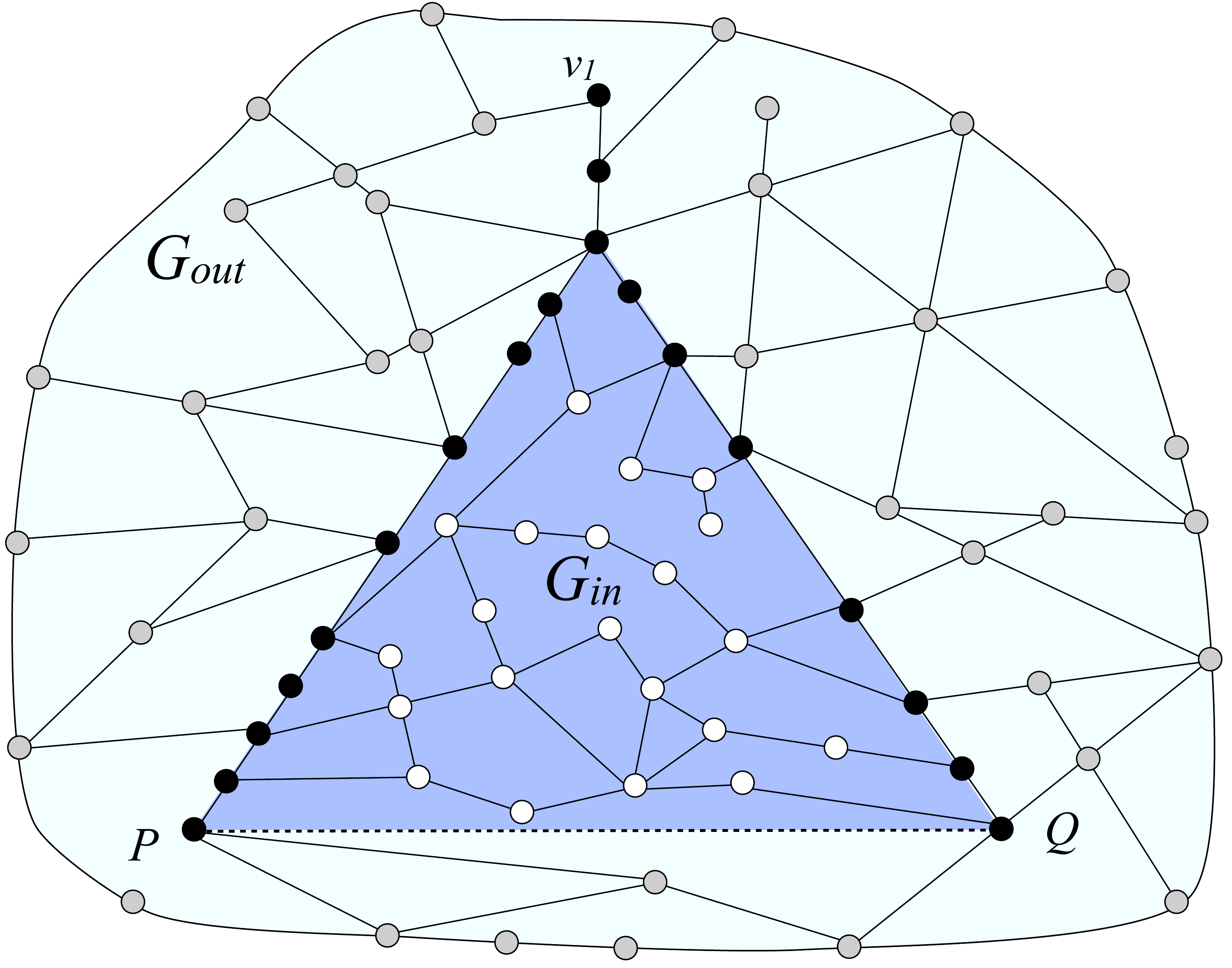}\ \ \ \ \ \
    \includegraphics[scale=0.18]{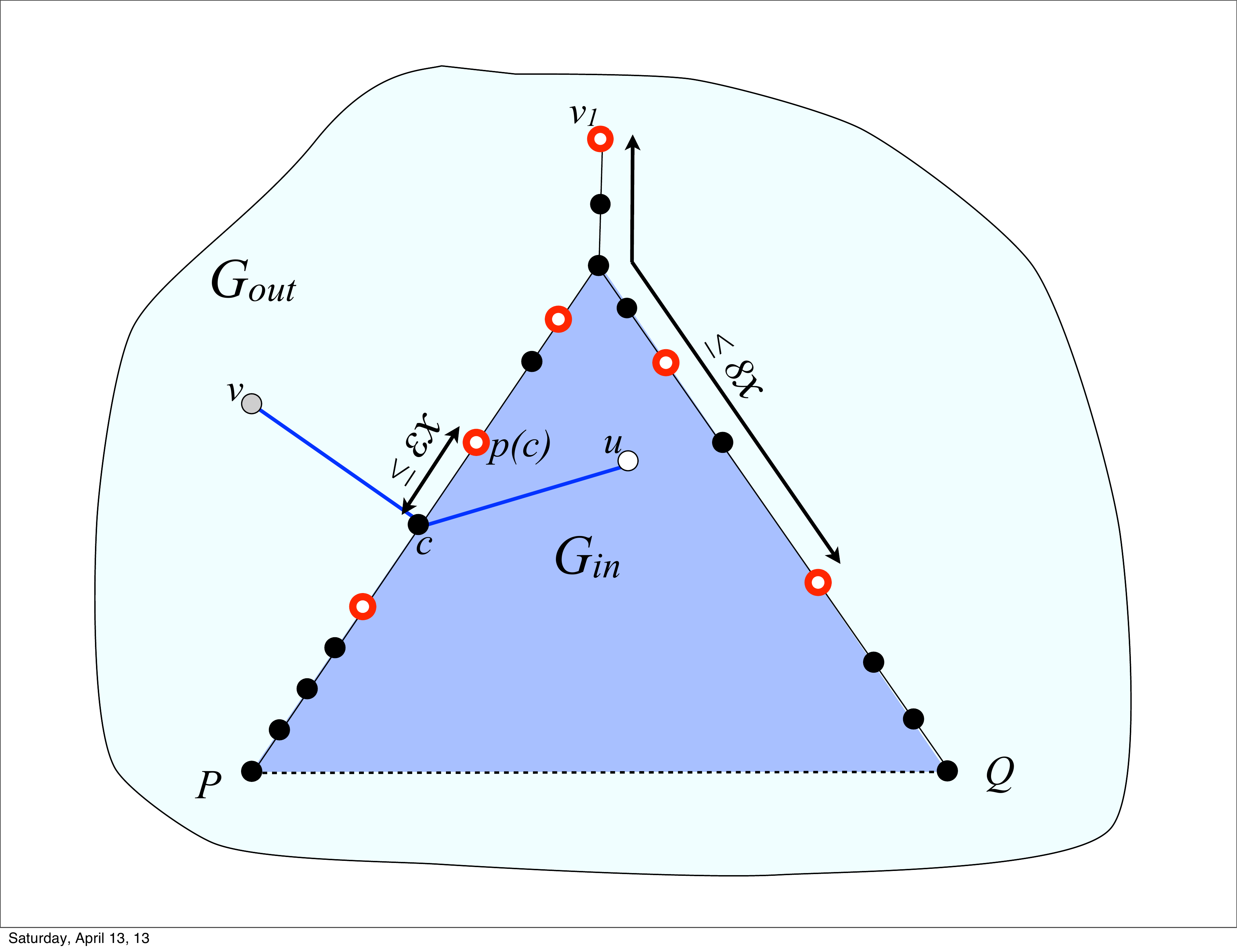}
   \caption{Two illustrations of a weighted undirected planar graph $G$. On the left: The black nodes constitute the shortest path separator $C$ composed of two shortest paths $P$ and $Q$ emanating from the same vertex $v_1$. The subgraph of $G$ induced by the white (resp. gray) nodes is denoted $A$ (resp. $B$). The graph $G_{in}$ (resp. $G_{out}$) is the subgraph induced by $A\cup C$ (resp. $B\cup C$). On the right: The six circled vertices are the $16/\varepsilon$ portals in the $8\x$ prefixes of $P$ and $Q$. The shortest path between $u$ and $v$ goes through the separator vertex $c$ and is approximated by the $u$-to-$\lambda(c)$ and the $\lambda(c)$-to-$v$ shortest paths where $\lambda(c)$ is the closest portal to $c$. The distance from $c$ to $\lambda(c)$ is at most $\varepsilon x$. }
  \label{separator}
 \end{figure}

Once we have chosen the portals, we move on to construct a tripartite graph $G_t$ whose three vertex sets (or columns) are $(V(G_{in}), portals, V(G_{out}))$. The length of edge $(u\in V(G_{in}),v\in portals)$ or $(u\in portals,v\in V(G_{out}))$ is the $u$-to-$v$ distance in $G$. This distance is computed by running the linear-time SSSP algorithm of Henzinger et al.~\cite{HKRS97} in $G$ from each of the $16/\varepsilon$ portals in total $O(1/\varepsilon\cdot |V(G)|)$ time. The following lemma states that our choice of portals implies that $d(G_{in},G_{out},G_t)$ is a good approximation of $d(G_{in},G_{out},G)$. 

\begin{lemma}\label{lemma:G_t}
If $d(G_{in},G_{out},G) \ge \x$, then 
$d(G_{in},G_{out},G_t)$  is  a $(1+2\varepsilon)$-approximation of $d(G_{in},G_{out},G)$. Otherwise, $d(G_{in},G_{out},G_t)\le (1+2\varepsilon) \x$.
\end{lemma}
\begin{proof}

The first thing to notice is that $d(G_{in},G_{out},G_t)\ge d(G_{in},G_{out},G)$. This is because every shortest $u$-to-$v$ path in $G_t$ between a marked vertex  $u\in V(G_{in})$ of the first column and a marked vertex $v\in V(G_{out})$ of  the third column corresponds to an actual $u$-to-$v$ path in $G$.

We now show that $d(G_{in},G_{out},G_t)\le (1+2\varepsilon) \cdot d(G_{in},G_{out},G)$. 
We begin with some notation. Let $P_t$ denote the shortest  path in $G_t$ realizing $d(G_{in},G_{out},G_t)$. 
The path $P_t$ is a shortest $u$-to-$v$ path for some {\em marked} vertices $u\in G_{in}$ and $v\in G_{out}$. The length of the path $P_t$ is  $\delta_{G_t}(u,v)$.
Let $P_G$ denote the shortest  $u$-to-$v$ path in $G$ that is of length  $\delta_{G}(u,v)$ and let $P_{\G}$ denote the shortest  $u$-to-$v$ path in the original graph $\G$ that is of length  $\delta_{\G}(u,v)$. Recall that we have the invariant that in every recursive level for every pair of marked vertices $\delta_G(u,v) \le (1+\varepsilon)\cdot \delta_{\G}(u,v)$. We also have that  $\delta_{\G}(u,v)\le 2\x$ and so  $\delta_G(u,v) \le  2\x\cdot (1 +\varepsilon) $. For the same reason, since $v_1$ (the first vertex of both $P$ and $Q$) is also marked, we know that $\delta_G(v_1,u)$ is  of length at most $2\x\cdot (1+\varepsilon)$. 

The path $P_G$ must include at least one vertex $c \in C$. Assume without loss of generality that $c \in P$. We claim that $c$ must be a vertex in the prefix of $P$ of length $8\x$. Assume the converse, then the $v_1$-to-$c$ prefix of $P$ is of length at least $8\x$. Since $P$ is a shortest path in $G$, this means that $\delta_G(v_1,c)$ is at least $8\x$.
However, consider the  $v_1$-to-$c$ path composed of the $v_1$-to-$u$ shortest path (of length $\delta_G(v_1,u)\le 2\x\cdot (1+\varepsilon)$) concatenated with the $u$-to-$c$ shortest path (of length $\delta_G(u,c)\le \delta_G(u,v) \le  2\x \cdot (1 +\varepsilon)$). Their total length is 
$4\x\cdot (1+\varepsilon)$ which is less than $8\x$ (since $\varepsilon <1$) thus contradicting our assumption.

After establishing that $c$ is somewhere in the $8\x$ prefix of $P$, we now want to show that $\delta_{G_t}(u,v)\le (1 +2\varepsilon)\cdot \delta_{G}(u,v)$. Let $\lambda(c)$ denote a closest portal to $c$ on the path $P$. Notice that by our choice of portals and since $c$ is in the $8\x$ prefix of $P$ we have that $\delta_G(c,\lambda(c)) \le \varepsilon \x$. By the triangle inequality we know that $\delta_G(u,\lambda(c)) \le \delta_G(u,c) + \delta_G(c,\lambda(c)) \le \delta_G(u,c) + \varepsilon \x$ and similarly $\delta_G(\lambda(c),v)  \le \delta_G(c,v) + \varepsilon \x$. This means that 
\begin{align*}
d(G_{in},G_{out},G_t) &= \delta_{G_t}(u,v)  \\
&\le \delta_{G}(u,\lambda(c)) + \delta_{G}(\lambda(c),v) \\ 
&\le  \delta_G(u,c) + \delta_G(c,v) + 2\varepsilon \x \\
&=   \delta_G(u,v) + 2\varepsilon \x \\
&\le d(G_{in},G_{out},G) + 2\varepsilon  \x \\
&\le  (1+2\varepsilon) \cdot d(G_{in},G_{out},G),
\end{align*}

\noindent 
where in the last inequality we assumed that $d(G_{in},G_{out},G) \ge \x$. Note that if $d(G_{in},G_{out},G) < \x$, then $d(G_{in},G_{out},G_t) \le (1+2\varepsilon) \cdot \x$. The lemma follows. \qed
\end{proof}

By Lemma~\ref{lemma:G_t}, approximating $d(G_{in},G_{out},G)$ when $d(G_{in},G_{out},G) \ge \x$ reduces to approximating $d(G_{in},G_{out},G_t)$. 
The case of $d(G_{in},G_{out},G) < \x$ means that the diameter $d$ of the original graph $\G$ is {\em not} a  $(u\in G_{in})$-to-$(v\in G_{out})$ path.
This is because $d\ge \x > d(G_{in},G_{out},G) \ge d(G_{in},G_{out},\G)$. 
So $d$ will be approximated in a different recursive call (when the separator separates the endpoints of the diameter). In the meanwhile, we 
will get that $d(G_{in},G_{out},G_t)$ is  at most  $(1+2\varepsilon) \cdot \x$ and so it will not compete with the correct recursive call when taking the maximum.

\subsection{Approximate $d(G_{in},G_{out},G_t)$}\label{ss:gt}
In this subsection, we show how to approximate the diameter in the tripartite graph $G_t$. We give a $(1+ 2\varepsilon)$-approximation for $d(G_{in},G_{out},G_t)$. By the previous subsection, this means we have a 
$(1+2\varepsilon)(1+2\varepsilon) < (1+ 5\varepsilon)$-approximation for $d(G_{in},G_{out},G)$. From the invariant that distances in $G$ between marked vertices are a $(1+\varepsilon)$-approximation of these distances in the original graph $\G$, we get a $(1+5\varepsilon)(1+\varepsilon) < (1+ 7\varepsilon)$-approximation for $d(G_{in},G_{out},\G)$ in the original graph $\G$.

We now present our $(1+ 2\varepsilon)$-approximation for 
$d(G_{in},G_{out},G_t)$ in the tripartite graph $G_t$. Recall that $P_t$ denotes the shortest path in $G_t$ that realizes $d(G_{in},G_{out},G_t)$. By the definition of $G_t$, we know that the  path $P_t$ is composed of only two edges: (1) edge $(u,p)$ between a marked vertex $u$ of the first column (i.e., $u\in V(G_{in})$) and a vertex $p$ of the second column (i.e., $p$ corresponds to some portal in $G$). (2) edge $(p,v)$ between $p$ and a marked vertex $v$ of the third column (i.e., $v\in V(G_{out})$).

Let $X$ (resp. $Y$) denote the set of all edges in $G_t$ adjacent to marked vertices of the first (resp. third) column. Let 
$\ell$ denote the maximum edge-length over all edges in $X\cup Y$. Notice  that $\ell \le d(G_{in},G_{out},G_t) \le 2\ell$. We round up the lengths of all edges in $X\cup Y$ to the closest multiple of $\varepsilon \ell$. The rounded edge-lengths are thus all in $\{\varepsilon \ell, 2\varepsilon \ell,3\varepsilon \ell,\ldots, \ell\}$. We denote $G_t$ after rounding as $G'_t$. Notice that $d(G_{in},G_{out},G'_t)$ is a $(1+2\varepsilon)$-approximation of $d(G_{in},G_{out},G_t)$. This is because the path $P_t$ is of length at least $\ell$ and is composed of two edges, each one of them has increased its length by at most $\varepsilon \ell$.

We now show how to compute $d(G_{in},G_{out},G'_t)$ exactly in linear time. We first divide all the edge-lengths of $G'_t$ by $\varepsilon \ell$ and get that $G'_t$  has edge-lengths in $\{1, 2,3,\ldots, 1/\varepsilon\}$. After finding $d(G_{in},G_{out},G'_t)$  (which is now a constant) we simply multiply the result by $\varepsilon \ell$. The following lemma states that when the diameter is constant it is possible to compute it exactly in linear time. Note that we can't just use Eppstein's~\cite{Eppstein} linear-time algorithm because it works only on planar graphs and in our case we get a {\em non-planar} tripartite graph $G'_t$.

\begin{lemma}\label{lemma:G'_t}
$\!d(G_{in},G_{out},G'_t)$  can be computed in time $O(|V(G)| / \varepsilon +2^{O(1/\varepsilon)})$.
\end{lemma}
\begin{proof}
Recall that in $G'_t$ we denote the set of all edges adjacent to marked vertices of the first  and third column as $X$ and $Y$. The length of each edge in $X\cup Y$ is in $\{1, 2,\ldots, k\}$ where $k=1/\varepsilon$. The number of edges in $X$ (and similarly in $Y$) is at most $16k\cdot |V(G)|$. This is because the first column contains $|G_{in}| \le |V(G)|$ vertices and the second column contains $j\le 16k$ vertices $v_1, v_2,\ldots, v_j$ (the portals). 
For every marked vertex $v$ in the first (resp. third) column, we store a $j$-tuple $v_X$ (resp. $v_Y$) containing the edge lengths from $v$ to all vertices of the second column. In other words, the $j$-tuple $v_X = \langle \delta(v,v_1), \delta(v,v_2),\ldots, \delta(v,v_j) \rangle$ where  every $\delta(v,v_i) \in \{1, 2,\ldots, k\}$ is the length of the edge $(v,v_i)$. The total number of tuples is $O(k\cdot |V(G)|)$ but the total number of {\em different} tuples is only $t = k^{O(k)}$ since each tuple has $O(k)$ entries and each entry is in $\{1, 2,\ldots, k\}$.

We create two binary vectors $V_X$ and $V_Y$ each of length $t$. The $i$'th bit of $V_X$ (resp. $V_Y$) is 1 iff the $i$'th possible tuple exists as some  $v_X$ (reps. $v_Y$). Creating these vectors takes  $O( k\cdot |V(G)|)= O( |V(G)| / \varepsilon)$ time. Then, for every 1 bit in $V_X$ (corresponding to a tuple of vertex $u$ in the first column) and every 1 bit in $V_Y$ (corresponding to a tuple of vertex $v$ in the third column)  we compute the $u$-to-$v$ distance in $G'_t$ using the two tuples in time $O(16k)$. We then return the maximum of all such $(u,v)$ pairs. 
Notice that a 1 bit can correspond to several vertices that have the exact same tuple. We arbitrarily choose any one of these.  
There are $t$ entries in $V_X$ and $t$ entries in $V_Y$ so there are  
$O(t^2)$ pairs of 1 bits. Each pair is examined in $O(16k)$ time for a total of $O(k t^2)=k^{O(k)}$ time.

To complete the proof we now show that this last term $O(k t^2)$ is not only $k^{O(k)}$ but actually $2^{O(k)}$. For that we claim that the total number of different tuples is $t = 2^{O(k)}$. We assume for simplicity (and w.l.o.g.) that all portals  $v_1,\ldots,v_j$ are on the separator $P$.
We encode a $j$-tuple $v_X = \langle \delta(v,v_1),\ldots, \delta(v,v_j) \rangle$ by a $(2j-1)$-tuple $v'_X$: The first entry of $v'_X$ is $\delta(v,v_1)$. The next $j-1$ entries are $|\delta(v,v_{i+1}) - \delta(v,v_i)|$ for  $i=1,\ldots, j-1$. Finally, the last $j-1$ entries are single bits where the $i$'th bit is 1 if $\delta(v,v_{i+1}) - \delta(v,v_i) \ge 0$  and 0 if $\delta(v,v_{i+1}) - \delta(v,v_i) < 0$.

We will show that the number of different $(2j-1)$-tuples $v'_X$ is $2^{O(k)}$. There are $k$ options for the first entry of $v'_X$ and two options (0 or 1) for each of the last  $j-1$ entries. We therefore only need to show that there are at most $2^{O(k)}$ possible $(j-1)$-tuples $\langle a_1,a_2,\ldots,  a_{j-1} \rangle$ where $a_i = |\delta(v,v_{i+1}) - \delta(v,v_i)|$.  First notice that  since $\delta(v,v_{i+1})$ and $\delta(v,v_i)$ correspond to distances, by the triangle inequality we have $a_i= |\delta(v,v_{i+1}) - \delta(v,v_i)| \le \delta(v_i,v_{i+1})$. We also know that $\delta(v_1,v_j) \le 8\x/\varepsilon \ell$ since all portals lie on a prefix of $P$ of length at most $8\x$ and we scaled the lengths by dividing by $\varepsilon \ell$. We get that $\sum_{i=1}^{j-1}{ a_i} \le 8\x/\varepsilon \ell  \le 16k$. In the last inequality we used the fact that $x\le 2\ell$, if $x > 2\ell$, then we ignore this recursive call altogether (the diameter will be found in another recursive call).  
To conclude, observe that the number of possible vectors $\langle a_1,a_2,\ldots,  a_{j-1} \rangle$ where every $a_i$ is non-negative and $\sum{ a_i} \le 16k$ is at most $2^{O(k)}$.  
\qed
\end{proof}

\noindent 
To conclude, we have so far seen how to  obtain a $(1+ 5\varepsilon)$-approximation for $d(G_{in},G_{out},G)$ implying a $(1+ 7\varepsilon)$-approximation for $d(G_{in},G_{out},\G)$ in the original graph $\G$. The next step is to unmark all vertices of $C$ and move on to recursively approximate $d(G_{in},G_{in},G)$ (approximating $d(G_{out},G_{out},G)$ is done similarly).

 \subsection{Reduce  $d(G_{in},G_{in},G)$ 
to $d(G_{in},G_{in},G_{in}^+)$}\label{ss:gplus}

In this subsection we show how to recursively obtain a $(1+ 5\varepsilon)$-approximation of $d(G_{in},G_{in},G)$ and recall that this implies a $(1+ 7\varepsilon)$-approximation of $d(G_{in},G_{in},\G)$ in the original graph $\G$ since we will make sure to maintain our invariant that, at any point of the recursion,
distances between marked vertices are a $(1+\varepsilon)$-approximation of these distances in the original graph $\G$.

It is important to note that our desired construction can be obtained with similar guarantees using the construction of Thorup~\cite{Thorup04} for distance oracles. However, we present here a simpler construction than \cite{Thorup04} since, as apposed to distance oracles that require {\em all-pairs} distances, we can afford to only consider distances that are between $x$ and $2x$.

There are two problems with applying recursion to solve $d(G_{in},G_{in},G)$. The first is that $|V(G_{in})|$ can be as large as $|V(G)|$ and we need it to be  at most  $2/3\cdot |V(G)|$. We do know however that the number of  {\em marked} vertices in $V(G_{in})$ is at most $2/3 \cdot  |V(G)|$ . The second problem is that it is possible that the $u$-to-$v$ shortest path in $G$ for $u,v\in G_{in}$ includes vertices of $G_{out}$. This  only happens if the $u$-to-$v$ shortest path in $G$ is composed of a shortest $u$-to-$p$  path ($p\in P$) in $G_{in}$, a  shortest $p$-to-$q$  path ($q\in Q$) in $G_{out}$, and a shortest $q$-to-$v$  path in $G_{in}$.
To overcome these two problems, we construct a planar graph $G_{in}^+$ that 
has at most $2/3  \cdot  |V(G)|$ vertices and $d(G_{in},G_{in},G_{in}^+)$  is a 
 $(1+\varepsilon/(2\log n))$-approximation to {$d(G_{in},G_{in},G)$. 
 
Recall that the subgraph $B$ of $G$ induced by all vertices in the strict exterior of the separator $C$  is such that $|B|\le 2/3 \cdot  |V(G)|$ and $G_{out}= B\cup C$. 
The construction of $G_{in}^+$ is done in two phases. In the first phase, we replace the  $B$ part of $G$ with a graph $B'$ of polylogarithmic size. In the second phase, we contract the $C$ part of $G$ to polylogarithmic size.
 
\paragraph{\bf Phase I: replacing ${\boldsymbol B}$ with ${\boldsymbol B'}$.}
To construct $G_{in}^+$, we first choose a subset of $256  \log n /\varepsilon$ vertices from $C$ called {\em dense portals}. The dense portals are chosen similarly to the regular portals but there are more of them.
The marked vertex $v_1$ (the first vertex of both $P$ and $Q$) is chosen as the first dense portal. Then, for $i=2,\ldots,128\log n/\varepsilon$  we start from $v_{i-1}$ and walk on $P$ until we reach the first vertex whose distance from $v_{i-1}$ via $P$ is greater than $\varepsilon \x/(16\log n)$. We set this vertex as the dense portal $v_i$ and continue to $i+1$. We do the same for $Q$, for a total  of $256\log n/\varepsilon$ dense portals. 

After choosing the dense portals, we compute all  $O((256\log n/\varepsilon)^2)$ shortest paths in $G_{out}$ between dense portals. This can be done using SSSP from each portal   in total $O(|V(G_{out})|\cdot\log n/\varepsilon)$ time. It can also be done using the Multiple Source Shortest Paths (MSSP) algorithm of Klein~\cite{Klein05} in total $O(|V(G_{out})|\cdot\log n + \log^2 n/\varepsilon^2)$ time. 

Let $B'$ denote the graph obtained by the union of all these dense portal to dense portal paths in $G_{out}$. 
Notice that since these are shortest paths, and since we assumed shortest paths are unique, then every two paths can share at most one consecutive subpath. The endpoints of this subpath are of degree $>2$. There are only $O((256\log n/\varepsilon)^2)$ paths so this implies that 
the graph $B'$ has at most $O((256\log n/\varepsilon)^4)$ vertices of degree $> 2$. 
We can therefore contract vertices of degree $= 2$. The number of vertices of $B'$ then decreases to $O((256\log n/\varepsilon)^4)$, it remains a planar graph, and its edge lengths correspond to subpath lengths. 

We then unmark all vertices of $B'$ and append $B'$ to the infinite face of $G_{in}$. In other words, we take the disjoint union of $G_{in}$ and $B'$ and
identify the dense portals of $G_{in}$ with the dense portals of $B'$.
This results in a graph $G_{in}^+$ that has $|V(G_{in})|+O((256\log n/\varepsilon)^4)$ vertices. In Lemma~\ref{lemma:G_in^+} we will show that $d(G_{in},G_{in},G_{in}^+)$ can serve as a $(1+\varepsilon/(2\log n))$-approximation to $d(G_{in},G_{in},G)$. But first we will shrink $G_{in}^+$ so that the number of its vertices is bounded by $2/3 \cdot  |V(G)|$.

\begin{figure}[h!]
   \centering
   \includegraphics[scale=0.19]{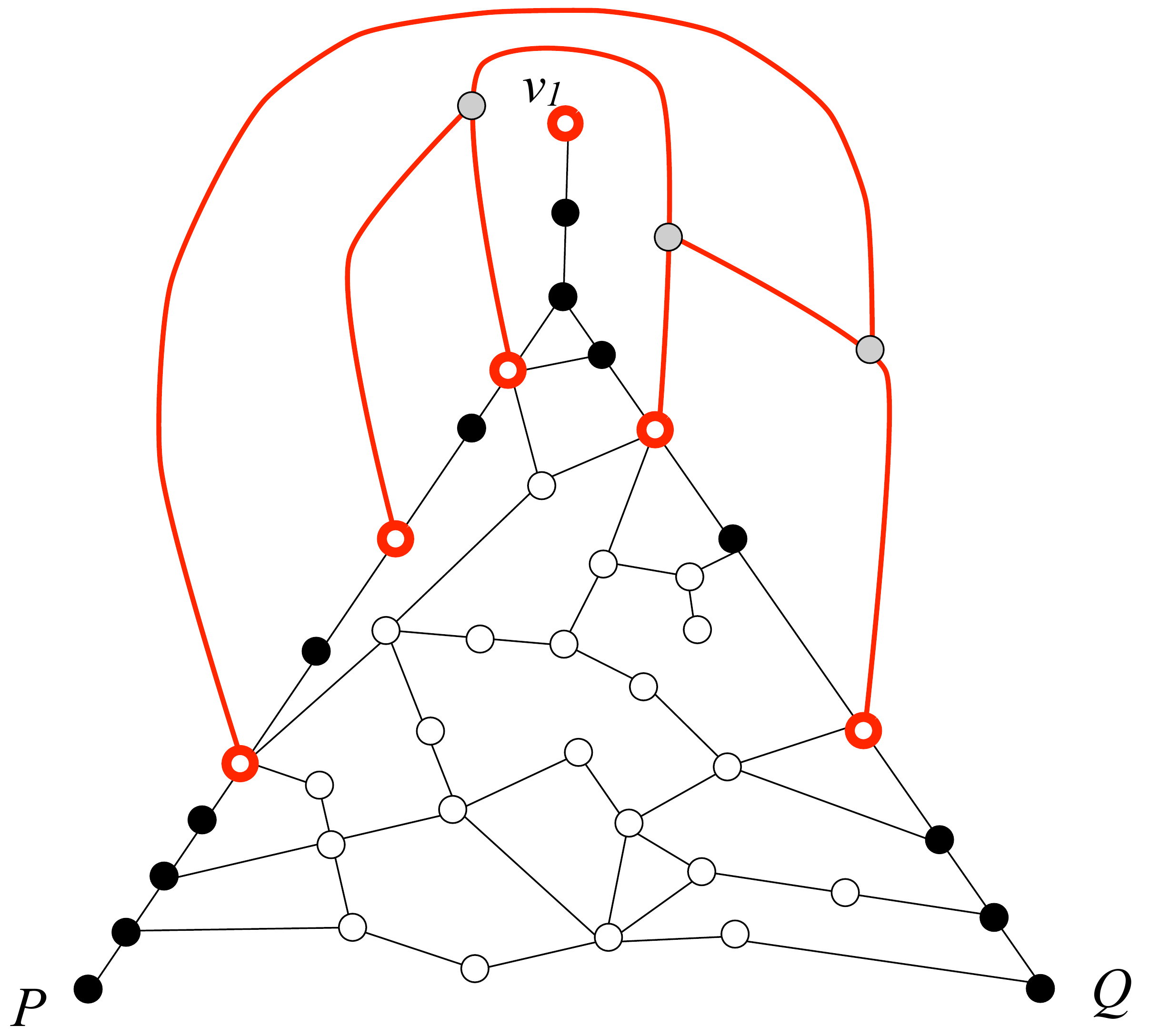}\ \ \ \ \ \ \ \ \ \ \ \ \ \ \ \ \ \ 
    \includegraphics[scale=0.19]{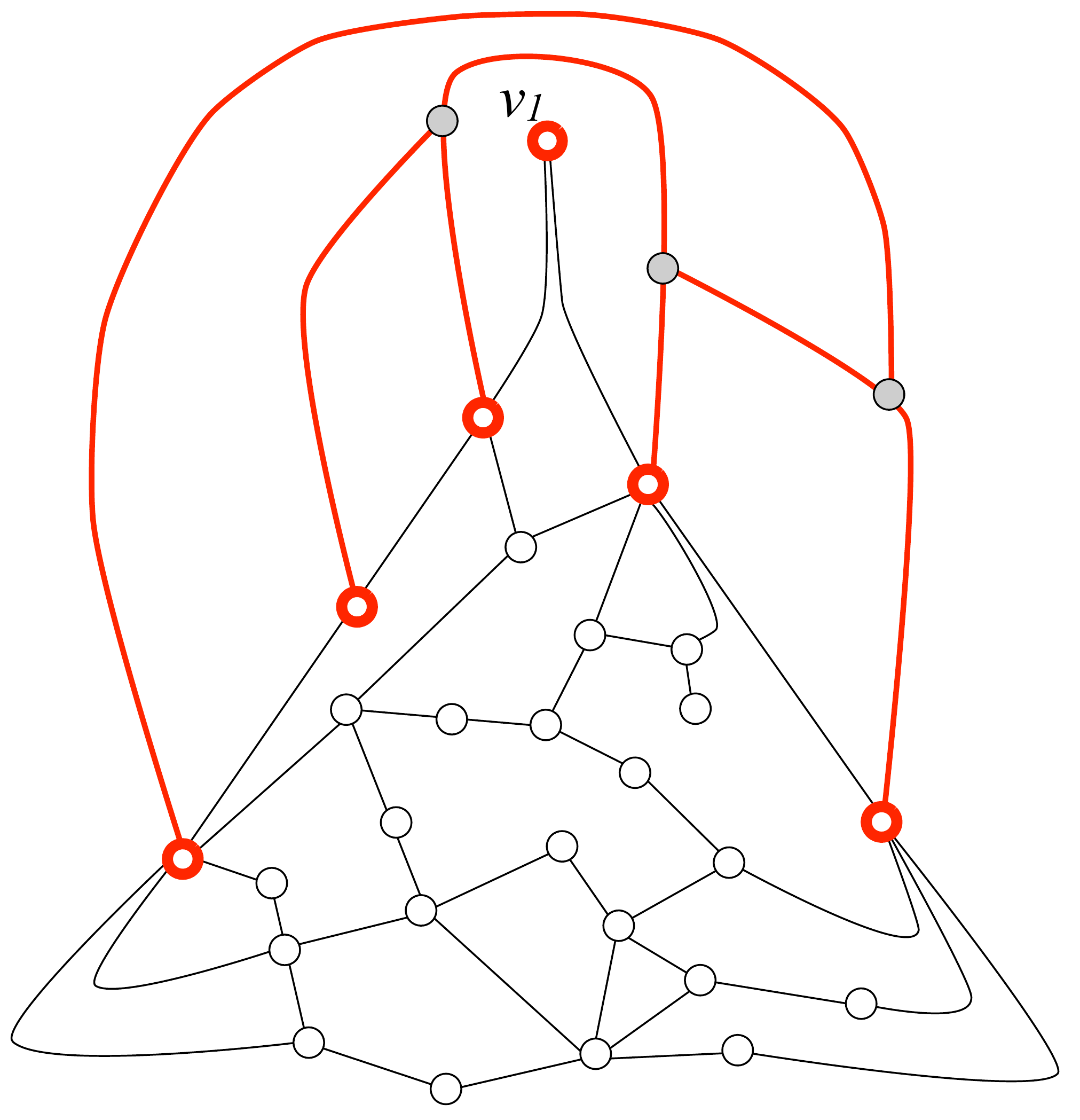}
   \caption{On the left: The graph $G_{in}^+$ before shrinking.   
   The white vertices are the vertices of $A$, the black vertices are the vertices of $C$ that are not dense portals, the six red circled vertices are the $256\log n/\varepsilon$ dense portals, and the gray vertices are the vertices of $B'\setminus C$ with degree $>2$.  On the right: The graph $G_{in}^+$ after shrinking. 
  The edges adjacent to vertices of $C$ that are not dense portals are now replaced with edges to dense portals.}
 \end{figure}

\paragraph{\bf Phase II: shrinking ${\boldsymbol G_{in}^+}$.}
The problem with the current $G_{in}^+$ is still that the size of $V(G_{in}^+)$ is not necessarily bounded by $2/3 \cdot  |V(G)|$. This is because $C$ (that is part of $V(G_{in}^+)$) can be as large as $n$. 
We now show how to shrink $V(G_{in}^+)$ to size  $2/3 \cdot  |V(G)|$ while $d(G_{in},G_{in},G_{in}^+)$ remains a $(1+\varepsilon/(2\log n))$-approximation of $d(G_{in},G_{in},G)$. 
To achieve this, we shrink the $C$ part of $V(G_{in}^+)$ so that it only includes the dense portals. We show how to shrink $P$, shrinking $Q$ is done similarly.

Consider two dense portals $v_i$ and $v_{i+1}$ on $P$ (i.e., $v_i$ is the closest portal to $v_{i+1}$ on the path $P$ towards $v_1$). 
We want to eliminate all vertices of $P$ between $v_i$ and $v_{i+1}$. Denote these vertices by $p_1,\ldots,p_k$. If $v_i$ is the last portal of $P$ (i.e., $i=128\log n$), then $p_1,\ldots,p_k$ are all the vertices between $v_i$ and the end of $P$.
Recall that $A$ is the subgraph of $G$ induced by all vertices in the strict interior of the separator $C$.
Fix a planar embedding of $G_{in}^+$.
We perform the following process as long as there is some vertex $u$ in $Q \cup A$ which is a neighbor of some $p_j$, and which is on some face
of the embedding that also contains $v_i$.
We want to ``force'' any shortest path that goes through an edge $(u,p_j)$ to also go through the dense portal $v_i$.
To this end, we delete all such edges $(u,p_j)$, and instead insert a single edge $(u,v_i)$ of length $\min_j\{ \ell(u,p_j) +\delta_G(p_j,v_i)\}$.
Here, $\ell(u,p_j)$ denotes the length of the edge $(u,p_j)$ (it may be that $\ell(u,p_j)=\infty$ if $(u,p_j)$ is not an edge) and $\delta_G(p_j,v_i)$
denotes the length of the $p_j$-to-$v_i$ subpath of $P$.
It is important to observe that the new edge $(u,v_i)$ can be embedded while maintaining the planarity since we have chosen $u$ to be on the same face
as $v_i$. Observe that once the process ends, the vertices $p_j$ have no neighbors in $Q \cup A$.

Finally, we replace the entire $v_{i+1}$-to-$v_i$ subpath of $P$ with a single edge  $(v_{i+1},v_i)$ whose  length is equal to  the entire subpath length. If $v_i$ is the last dense portal in $P$, then we simply delete the entire subpath between $v_i$ and the end of $P$.
The entire shrinking process takes only linear time in the size of $|V(G)|$ since it is linear in the number of edges of $G_{in}^+$
(which is a planar graph).

The following Lemma asserts that after the shrinking phase $d(G_{in},G_{in},G_{in}^+)$ can serve as a $(1+\varepsilon/(2\log n))$-approximation to $d(G_{in},G_{in},G)$.

\begin{lemma}\label{lemma:G_in^+}
$d(G_{in},G_{in},G) \le d(G_{in},G_{in},G_{in}^+) \le d(G_{in},G_{in},G) + \varepsilon  \x /(2\log n)$
\end{lemma}
\begin{proof}

First observe that $d(G_{in},G_{in},G_{in}^+) \ge d(G_{in},G_{in},G)$. This is because every vertex of $G_{in}$ that is marked in $G$ is also a marked vertex in $G_{in}^+$, and any shortest $u$-to-$v$ path in $G_{in}^+$ corresponds to an actual $u$-to-$v$ path in $G$.

We now show that $d(G_{in},G_{in},G_{in}^+) \le d(G_{in},G_{in},G) + \varepsilon  \x /(2\log n)$. 
Let $P^+$ denote the shortest $u$-to-$v$  path in $G_{in}^+$ realizing $d(G_{in},G_{in},G_{in}^+)$. 
Both $u$ and $v$ are  marked vertices in $G_{in}$ and the length of $P^+$ is  $\delta_{G_{in}^+}(u,v)$.
Let $P_G$ denote the shortest  $u$-to-$v$ path in $G$ that is of length  $\delta_{G}(u,v)$. \vspace{0.1in}

\noindent {\bf Case 1:} If $P_G$ does not include any vertex of $C$, then $P_G$ is also present in $G_{in}^+$ and therefore $d(G_{in},G_{in},G_{in}^+) \le d(G_{in},G_{in},G)$. \vspace{0.1in}

\noindent{\bf Case 2:} If $P_G$ includes vertices that are not in $G_{in}$ (i.e., vertices in $G_{out}\setminus C$), then $P_G$ must be composed of a  shortest $u$-to-$p$  path ($p\in P$) in $G_{in}$, a  shortest $p$-to-$q$  path ($q\in Q$) in $G_{out}$, and a shortest $q$-to-$v$  path in $G_{in}$.

We first claim that $p$ must be a vertex in the prefix of $P$ of length $8\x$ (a similar argument holds for $q$ and $Q$). Assume the converse, then the  prefix of $P$ from $v_1$ (the first vertex of both $P$ and $Q$) to $p$ is of length at least $8\x$. 
Recall that we have the invariant that in every recursive level for every pair of marked vertices $\delta_G(u,v) \le (1+\varepsilon)\cdot \delta_{\G}(u,v)\le 2\x\cdot (1+\varepsilon)$. For the same reason we know that $\delta_G(v_1,u) \le 2\x\cdot (1+\varepsilon)$. 
Since $P$ is a shortest path in $G$, this means that $\delta_G(v_1,p) \ge 8\x$.
However, consider the  $v_1$-to-$p$ path composed of the $v_1$-to-$u$ shortest path (of length $\delta_G(v_1,u)\le 2\x\cdot (1+\varepsilon)$) concatenated with the $u$-to-$p$ shortest path (of length $\delta_G(u,p)\le \delta_G(u,v) \le  2\x \cdot (1 +\varepsilon)$). Their total length is 
$4\x\cdot (1+\varepsilon)$ which is less than $8\x$ (since $\varepsilon <1$) thus contradicting our assumption.

We now show that $\delta_{G_{in}^+}(u,v)\le \delta_{G}(u,v) +\varepsilon \x/(2\log n)$. For $c\in P$ (resp. $c\in Q$), let $\lambda(c)$ denote the first dense portal encountered while walking from $c$ towards $v_1$ on the path $P$ (resp. $Q$). Notice that since $p$ and $q$ are in the $8\x$ prefixes of $P$ and $Q$  we have that $\delta_G(p,\lambda(p)) \le \varepsilon \x/(16\log n)$ and $\delta_G(q,\lambda(q)) \le \varepsilon \x/(16\log n)$. 
From the shrinking phase, it is easy to see that $G_{in}^+$ includes a $u$-to-$\lambda(p)$ path of length 
$\delta_G(u,p) + \delta_G(p,\lambda(p))$ and so 
$\delta_{G_{in}^+}(u,\lambda(p)) \le \delta_G(u,p) + \varepsilon \x/(16\log n)$. Similarly,  $\delta_{G_{in}^+}(\lambda(q),v)  \le \delta_G(q,v) + \varepsilon \x/(16\log n)$. Furthermore, since $G_{in}^+$ was appended with shortest paths  between dense portals in $G_{out}$ we have  
 $ \delta_{G_{in}^+}(\lambda(p),\lambda(q))\le  \delta_{G_{out}}(\lambda(p),p) +  \delta_{G_{out}}(p,q) + \delta_{G_{out}}(q,\lambda(q)) =\delta_{G}(\lambda(p),p) +  \delta_{G_{out}}(p,q) + \delta_{G}(q,\lambda(q)) \le  \delta_{G_{out}}(p,q) + \varepsilon \x/(8\log n)$. 
To conclude we get that 
\begin{align*}
d(G_{in},G_{in},G_{in}^+) &= \delta_{G_{in}^+}(u,v)  \\
&\le \delta_{G_{in}^+}(u,\lambda(p)) + \delta_{G_{in}^+}(\lambda(p),\lambda(q))  + \delta_{G_{in}^+}(\lambda(q),v) \\ 
&\le  \delta_G(u,p) +  \delta_{G_{out}}(p,q) + \delta_G(q,v) + \varepsilon \x /(4\log n)\\
&=   \delta_G(u,v) + \varepsilon \x/(4\log n) \\
&\le d(G_{in},G_{in},G) + \varepsilon  \x /(4\log n)\\
&< d(G_{in},G_{in},G) + \varepsilon  \x /(2\log n).
\end{align*}

\noindent {\bf Case 3:} Finally, we need to consider the case where $P_G$ includes only vertices of $G_{in}$. We assume  $P_G$ includes vertices of $P$ and$\setminus$or 
vertices of $Q$ (otherwise this was handled in Case 1). We focus on the case that  $P_G$ includes vertices of both  $P$ and $Q$. The case that  $P_G$ includes vertices of one of $P$ or $Q$ follows immediately using a similar argument. 

Since $P$ and $Q$ are shortest paths, then $P_G$ must be composed of the following shortest paths: a $u$-to-$p$ path ($p\in P$) in $G_{in}$, a $p$-to-$p'$ subpath ($p'\in P$)  of $P$, a $p'$-to-$q'$ path ($q'\in Q$) in $G_{in}$, a $q'$-to-$q$  subpath ($q\in Q$) of $Q$, and a $q$-to-$v$ path in $G_{in}$.
Following the same argument as in Case 2, we know that $p$ and $p'$ (resp. $q$ and $q'$) must in the prefix of $P$ (resp. $Q$) of length $8\x$. This means $\delta_G(c,\lambda(c)) \le \varepsilon \x/(16\log n)$ for every $c\in \{p,p',q,q'\}$.

From the shrinking phase, it is easy to see that $G_{in}^+$ includes a $u$-to-$\lambda(p)$ path of length 
$\delta_G(u,p) + \delta_G(p,\lambda(p))$ and so 
$\delta_{G_{in}^+}(u,\lambda(p)) \le \delta_G(u,p) + \varepsilon \x/(16\log n)$. Similarly, we have  that  $\delta_{G_{in}^+}(\lambda(q),v)  \le \delta_G(q,v) + \varepsilon \x/(16\log n)$, and we have  that  $\delta_{G_{in}^+}(\lambda(p'),\lambda(q'))  \le  \delta_G(\lambda(p'), p') +  \delta_G(p', q') +\delta_G(q', \lambda(q')) \le \delta_G(p', q') + \varepsilon \x/(8\log n)$.
Furthermore, since subpaths of $P$ in $G_{in}^+$ between dense portals capture their exact distance in $G$ we have that $\delta_{G_{in}^+}(\lambda(p),\lambda(p')) \le \delta_G(\lambda(p),p) +\delta_G(p, p') + \delta_G(p', \lambda(p')) \le \delta_G(p, p') + \varepsilon \x/(8\log n)$ and similarly  $\delta_{G_{in}^+}(\lambda(q'),\lambda(q)) \le \delta_G(q', q) + \varepsilon \x/(8\log n)$.
To conclude we get that 
\begin{align*}
& d(G_{in},G_{in},G_{in}^+) = \delta_{G_{in}^+}(u,v)  \\
&\le \delta_{G_{in}^+}(u,\! \lambda(p)) \!+\! \delta_{G_{in}^+}(\lambda(p),\! \lambda(p'))  \!+\! \delta_{G_{in}^+}(\lambda(p'),\!\lambda(q')) \!+\! \delta_{G_{in}^+}(\lambda(q'),\! \lambda(q))\!+\! \delta_{G_{in}^+}(\lambda(q),\! v)   \\ 
&\le  \delta_G(u,p) +\delta_G(p, p') + \delta_G(p', q') + \delta_G(q', q) + \delta_G(q,v) + \varepsilon \x/(2\log n) \\
&=   \delta_G(u,v) + \varepsilon \x/(2\log n) \\
&< d(G_{in},G_{in},G) + \varepsilon  \x /(2\log n).
\end{align*}
\qed
\end{proof}

\begin{corollary}\label{col}
If $d(G_{in},G_{in},G)\ge \x$, then 
$d(G_{in},G_{in},G_{in}^+) $ is a $(1+\varepsilon/(2\log n))$-approximation  of $d(G_{in},G_{in},G)$.
If  $d(G_{in},G_{in},G) < \x$, then  $d(G_{in},G_{in},G_{in}^+) \le (1+\varepsilon/(2\log n)) \cdot \x$. 
\end{corollary}

\noindent By the above corollary, approximating $d(G_{in},G_{in},G)$ when $d(G_{in},G_{in},G) \ge \x$ reduces to approximating $d(G_{in},G_{in},G_{in}^+)$. 
When $d(G_{in},G_{in},G) < \x$ it means that the diameter of the original graph $\G$ is {\em not} a  $(u\in G_{in})$-to-$(v\in G_{in})$ path and will thus be approximated in a different recursive call.

Finally, notice that indeed we maintain the invariant that the distance between any two {\em marked} vertices in the recursive call to $G_{in}^+$  is a $(1+\varepsilon)$-approximation of the distance in the original graph $\G$. This is because, by the above corollary, every recursive call adds a $1+\varepsilon/(2\log n)$
factor to the approximation. Each recursive call decreases the input size by a factor of $(2/3+o(1))^{-1}$.
Hence, the overall depth of the recursion is at most $\log_{1.5-o(1)}n < 1.8 \log n$. Since
$$
(1+\varepsilon/(2 \log n))^{1.8 \log n} < e^{0.9 \varepsilon} < 1+ \varepsilon
$$
the invariant follows (we assume in the last inequality that  $\varepsilon \le 0.1$). 
Together with the $(1+ 5\varepsilon)$-approximation for $d(G_{in},G_{out},\G)$ in the original graph $\G$, we get a
$(1+5\varepsilon)\cdot (1+\varepsilon)\le (1+7\varepsilon)$-approximation of $d(G_{in},G_{in},\G)$ in the original graph $\G$, 
once we apply recursion to $d(G_{in},G_{in},G_{in}^+)$.

We note that our recursion halts once $|G_{in}^+|\le (256\log n /\varepsilon)^4$ in which case we naively compute $d(G_{in},G_{in},G_{in}^+)$ using APSP in time $O(|G_{in}^+|^2)$.
Recall that even at this final point, the distances between marked vertices still obey the invariant.

\subsection{Running time}
We now examine the total running time of our algorithm.
Let $n$ denote the number of vertices in our original graph $\G$ and let $V(G)$ denote the vertex set of the graph $G$ in the current invocation of
the recursive algorithm. The current invocation approximates $d(G_{in},G_{out},G_t)$ as shown in subsection \ref{ss:gt} in time
$O(|V(G)| / \varepsilon +2^{O(1/\varepsilon)})$.
It then goes on to construct the subgraphs $G_{in}^+$ and $G_{out}^+$ as shown in subsection \ref{ss:gplus}, where we have that
after contraction using dense portals, $|V(G_{in}^+)| = \alpha |V(G)|+O(\log^4 n/\varepsilon^4)$ and $|V(G_{out}^+)| = \beta|V(G)|+O(\log^4 n/\varepsilon^4)$,
where $\alpha,\beta \le 2/3$ and $\alpha+\beta \le 1$. The time to construct $|V(G_{in}^+)|$ and $|V(G_{out}^+)|$ is dominated by the time required to compute SSSP for each dense portal,
which requires $O(|V(G)|\cdot\log n/\varepsilon)$. We then continue recursively to $G_{in}^+$ and to $G_{out}^+$.
Hence, if $T(|V(G)|)$ denotes the running time for $G$, then we get that 
\begin{eqnarray*}
T(|V(G)|) & = & O\big(|V(G))|\cdot\log n/\varepsilon +2^{O(1/\varepsilon)}\big)\\
& + & T\big(\alpha|V(G)|+O(\log^4 n/\varepsilon^4)\big)\\
& + & T\big(\beta|V(G)|+O(\log^4 n/\varepsilon^4)\big).
\end{eqnarray*}

\noindent In the recursion's halting condition, once we get to components of size $|V(G)| = (256\log n /\varepsilon)^4$, we naively run APSP. This takes $O(|V(G)|^2)$ time for each such component, and there are $O(n/|V(G)|)$ such components, so the total time is $O(n\cdot|V(G)|)= O(n\log^4 n /\varepsilon^4)$. 
It follows that
$$
T(n) = O(  n \log^4 n/\varepsilon^4 +n\cdot 2^{O(1/\varepsilon)}).
$$

\section{Concluding Remarks}

We presented the first $(1+\varepsilon)$-factor approximation algorithm for the diameter of an undirected planar graph with non-negative edge lengths.
Moreover, it is the first algorithm that provides a nontrivial (i.e. less than $2$-factor) approximation in near-linear time.

It might still be possible to slightly improve the running time of our algorithm by removing a logarithmic factor, or by replacing the exponential dependency on $\varepsilon$
with a polynomial one. In addition, the technique of Abraham and Gavoille~\cite{AG06} which generalizes shortest-path separators to the class of H-minor free  graphs
may also turn out to be useful.
%
%

\bibliographystyle{plain}
\bibliography{planar-diam}

\end{document}